\documentclass[11pt,a4paper]{article}

\usepackage[T1]{fontenc}
\usepackage[utf8]{inputenc}
\usepackage{lmodern}
\usepackage{microtype}
\usepackage{amsmath,amssymb,amsthm,mathtools}
\usepackage{booktabs}
\usepackage{enumitem}
\usepackage{xcolor}
\usepackage[hidelinks]{hyperref}
\usepackage[nameinlink,capitalise]{cleveref}
\usepackage[a4paper,margin=30mm]{geometry}

\hypersetup{
  pdftitle={Substantive Agency and Computational Non-Anticipability: An Axiomatic Route to a Conditional Separation of P and NP},
  pdfauthor={Jérôme Clech},
  pdfsubject={Substantive agency, computability, and conditional search-complexity consequences},
  pdfkeywords={free will, agency, singularisation, computability, search complexity, FP, FNP, P versus NP}
}

\newtheorem{theorem}{Theorem}[section]
\newtheorem{proposition}[theorem]{Proposition}
\newtheorem{lemma}[theorem]{Lemma}
\newtheorem{corollary}[theorem]{Corollary}
\theoremstyle{definition}
\newtheorem{definition}[theorem]{Definition}
\newtheorem{countermodel}[theorem]{Countermodel}
\theoremstyle{remark}
\newtheorem{remark}[theorem]{Remark}

\newcommand{\GLUE}{\mathsf{GLUE}}
\newcommand{\SELECT}{\mathsf{SELECT}}
\newcommand{\FPS}{\mathrm{FP}_{\mathrm{search}}}
\newcommand{\FNPS}{\mathrm{FNP}_{\mathrm{search}}}
\newcommand{\TFNPS}{\mathrm{TFNP}_{\mathrm{search}}}
\newcommand{\RSEL}{R_{\mathsf{SEL}}}
\newcommand{\poly}{\mathrm{poly}}
\newcommand{\Aut}{\operatorname{Aut}}
\newcommand{\LAeff}{\mathrm{LA}_{\!\mathrm{eff}}}
\newcommand{\LAc}{\mathrm{LA}_{\!c}}
\newcommand{\LAnc}{\mathrm{LA}_{\!\mathrm{nc}}}
\newcommand{\LAo}{\mathrm{LA}_{\!o}}
\newcommand{\AIK}{\mathsf{AIK}}

\title{Substantive Agency and Computational Non-Anticipability:\\
An Axiomatic Route to a Conditional\\
Separation of $P$ and $NP$}

\author{J\'er\^ome Clech\\
\small Colonel, PhD, Habilitation to Supervise Research (HDR)\\
\small Chairholder, Chair of Applied Air and Space Strategies\\
\small Centre for Aerospace Strategic Studies, French Air and Space Force\\
\small Associate Researcher, Technology and Global Affairs Innovation Hub\\
\small Paris School of International Affairs (PSIA), Sciences Po\\
\small \texttt{jerome.clech@sciencespo.fr}}

\date{Version 0.17 -- 23 September 2026}

\begin{document}

\maketitle

\begin{abstract}
This paper characterises substantive agency and identifies the additional bridges under which it has a standard complexity-theoretic consequence: conditionally, $P\ne NP$. The constitutive core separates coherent plurality, causal openness, anticipatory non-pointing, act-level singularisation, and endogenous sourcehood. Its temporal claim yields a two-sided modal result: exact passive pre-act selection is incompatible with jointly retaining singularising priority and invariance under causally inert informational extension. If every standard polynomial procedure is deployable in that passive form, the actualisation map $\alpha$ has no polynomial-time selector. On an effective presentation, $\alpha$ is therefore either computable outside FP or noncomputable. A separate certified-actualisation protocol supplies polynomially bounded, polynomially verifiable traces with extensional unique projection. This is the constructive interface that turns historical actualisation into a total standard search relation; combined with the modal result, it places that relation in $\TFNPS\setminus\FPS$, from which a self-contained search argument yields $P\ne NP$. The result is a domain-relative conditional transfer theorem: it neither assumes polynomial non-anticipability in the definition of agency nor proves that human decision-making satisfies the bridge and certification premises.
\end{abstract}

\noindent\textbf{Keywords:} free will; agency; singularisation; sheaf theory; global selection; computability; search complexity; FP; FNP; $P$ versus $NP$.

\section{Introduction}

Free-will debates often move too quickly between three different questions. The first is conceptual: what distinguishes a free act from deterministic execution, chance, coercion, or mere retrospective attribution? The second is representational: how can that distinction be expressed in a mathematical structure without reducing it to a verbal stipulation? The third is computational: if the realised act can be verified after it occurs but cannot be selected efficiently beforehand, what follows for standard search complexity? These questions are connected, but they are not interchangeable.

Two public preprints address the computational and ontological ends of this chain \cite{Clech2026a,Clech2026b}. They are antecedents and extended technical records, not premises that the reader must accept. The computational preprint studies a finite choice structure in which local constraints determine a set of globally admissible continuations, denoted by $\GLUE$, while a distinct operation $\SELECT$ identifies the continuation actually realised. The ontological preprint studies the stronger case in which the realised sequence is causally incompressible and no uniform Turing predictor computes it from pre-act histories. The present paper restates and proves every bridge needed for its own main conclusion; it refers to the preprints only for fuller development of their respective models.

The present paper addresses the missing bridge. It does not assume that ``free will'' means polynomial non-anticipability. Causal non-determination excludes a prior determining mechanism occurrence by occurrence, regardless of cost. Actual anticipatory non-pointing follows from singularising priority under actual access-completeness. The computational question is asked only afterwards. The modal result is first an incompatibility: passive-extension invariance says that a causally inert read-only device cannot alter the substantive status of the act, whereas an exact output present inside the enlarged pre-act frame removes singularising priority. Effective deployment adds the further claim that any FP procedure can in principle be implemented in that passive form. The resulting FP exclusion and computable/noncomputable bifurcation are exhaustive only in this declared domain, while certified actualisation remains a further interface.

Two results organise the paper. First, the modal incompatibility theorem isolates the exact conflict among substantive singularising priority, passive-extension invariance, and deployable exact prediction; effective deployment then excludes every FP selector. Second, the certified-actualisation theorem shows that short verifiable traces with unique projection define a total search relation without a polynomial selector and therefore yield, conditionally, $P\ne NP$. The first result concerns the status of pre-act selection; the second supplies the independent bridge to standard complexity theory.

The central distinction is between \emph{admissibility} and \emph{actualisation}. Local and global constraints may determine a non-singleton set of coherent continuations. They do not thereby determine which continuation becomes historical. An act transforms an unpointed possibility structure into a pointed one. We call this transition \emph{singularisation}. The act has this substantive office only if it supplies the first exact point inside the relevant closed agency--information frame. If a coupled knower already possesses and can operationally supply the exact future point, the knower and decider need not be the same biological individual: they nevertheless belong to one informationally relevant composite system, whose pre-act state is already pointed. If the alleged foreknowledge is wholly isolated and unavailable to that system, it is not a member of the pre-act pointing class and does not violate non-pointing. Singularisation is indispensable, but it is not the sole basis of the account: without plurality there is nothing to singularise; without causal openness the issue is already determined; without anticipatory non-pointing the relevant frame is already pointed; without sourcehood the transition may be noise; without coherence the output need not count as a continuation of the prior situation.

A related representational point motivates, but does not prove, the later results. A physical situation becomes a formal problem through an operation that fixes a vocabulary, variables, relevance criteria, and an objective; a second operation produces finite instances; gluing identifies coherent continuations; and actualisation selects the historical one. Thus a TSP instance does not require a human operator once encoded, while the constitution and interpretation of applied problems remain agent-relative practices. No step of the complexity proof depends on that observation.

Our claim must be stated with care. No finite mathematical theorem can force every philosophical use of the expression ``free will'' to adopt one definition. A compatibilist may identify freedom with reasons-responsiveness or absence of direct coercion even under determinism \cite{Frankfurt1969,FischerRavizza1998}. The defensible claim is conditional but constraining: any conception that intends to preserve real alternatives, pre-act openness, the act as the event which first singularises one continuation, and the agent as the ultimate source of that transition must instantiate the structural core defined below. Rejecting the core requires abandoning at least one of those commitments, not merely changing terminology. The compatibilist analysis constrains that constitutive disagreement without supplying the later computational bridges: substantive non-reducibility yields the causal clause (A3-C), while singularising priority in a closed pre-act frame yields anticipatory non-pointing (A3-P). Passive deployment and certified actualisation remain independent, contestable additions.

The contribution is not any one of plurality, unpredictability, sourcehood, sheaf representation, or search complexity taken separately. Each has an established literature. It is the formally stratified bridge between them. More precisely, the paper contributes:
\begin{enumerate}[label=(\roman*)]
  \item an axiomatic separation of causal openness, anticipatory non-pointing, act-level singularisation, and endogenous sourcehood, with countermodels showing their distinct roles;
  \item a pre-pointing explication: in an actually access-complete frame, substantive singularising priority entails actual anticipatory non-pointing;
  \item a two-sided modal dilemma separating singularising priority, passive-extension invariance, and deployability of an exact selector, with the conceptual cost of each refusal made explicit and the FP exclusion stated as a corollary under effective deployment;
  \item a realisation-gap theorem showing that the computable branch does not by itself possess polynomially verifiable certificates;
  \item a certified-graph interface showing that short retrospective certificates with extensional unique projection collapse the bifurcation to the computable non-FP branch and generate a TFNP search relation with no polynomial selector;
  \item a self-contained conditional separation theorem, together with exact statements of what is constitutive, what is computational, and what remains an empirical correspondence hypothesis.
\end{enumerate}

The paper is organised accordingly. \Cref{sec:levels} separates the levels of analysis and \cref{sec:related} locates the differential contribution. \Cref{sec:axioms} gives the agential core, and \cref{sec:countermodels} tests the role of its clauses. \Cref{sec:compatibilism} isolates the compatibilist boundary. \Cref{sec:representation} gives the choice-frame and sheaf-theoretic representation. \Cref{sec:prediction} proves the modal result and the FP corollary. \Cref{sec:objectification} states the certification protocol and proves the realisation gap. \Cref{sec:interfaces} proves the exhaustive bifurcation and the conditional complexity consequence. \Cref{sec:scope} states the scope, correspondence conditions, and points of refusal.

\section{Five levels that must not be conflated}\label{sec:levels}

The argument has five levels, summarised in \cref{tab:levels}.

\begin{table}[ht]
\centering
\small
\begin{tabular}{p{0.19\textwidth}p{0.35\textwidth}p{0.34\textwidth}}
\toprule
Level & Principal question & What it does not yet supply \\
\midrule
Constitutive & Does the process exhibit plurality, causal openness, non-pointing, actualisation, and sourcehood? & No encoding, complexity bound, or certificate \\
Representational & Can that structure be expressed by choice frames and local-to-global sections? & No effective manipulation or lower bound \\
Effective presentation & Do occurrences and continuations have uniform finite encodings? & No short post-act certificates \\
Certified objectification & Are completed acts witnessed by short, polynomially verifiable traces with unique projection? & No selector and no unconditional class separation \\
Complexity & Does the resulting total search relation have a polynomial selector? & No direct metaphysical or empirical conclusion \\
\bottomrule
\end{tabular}
\caption{The five levels of the argument and their logical separation.}
\label{tab:levels}
\end{table}

This stratification blocks two common circularities. First, polynomial non-anticipability is not inserted into the constitutive definition. Second, polynomial traceability is not treated as a condition for an act to be free. Third, effective computability is not conflated with certified actualisation. These properties enter through distinct bridge principles.

\section{Related work and differential contribution}\label{sec:related}

The present programme meets several established literatures, but their conclusions should not be conflated.

\paragraph{Computational irreducibility and prediction.}
Computational irreducibility and no-shortcut arguments explain why the future behaviour of a deterministic process may be unavailable substantially earlier than the process itself produces it. Lloyd applies diagonal and time-complexity considerations to decision makers and explicitly presents the result as an explanation of the \emph{impression} of free will \cite{Lloyd2013}. Wolpert's inference-device framework proves law-independent limitations on the inferences available to physical devices embedded in one universe \cite{Wolpert2008}. These are genuine impossibility results, but neither identifies substantive agency with the first act-level pointing of one member of a coherent plurality, and neither supplies the certified unique-projection interface used below.

\paragraph{Computational sourcehood.}
Krumm and Mueller distinguish mere computational irreducibility from \emph{computational sourcehood}: roughly, reproducing a process's behaviour should require a structure-preserving representation of the relevant process itself \cite{KrummMueller2021}. Their deterministic, compatibilist analysis uses simulation preorders between Turing machines. The present criterion is not a refinement of that preorder. Functional ownership or unavoidable simulation is insufficient for pointwise (A3-C) when a complete prior mechanism fixes the act. Conversely, our substantive core alone supplies no complexity lower bound: the transfer requires separately auditable modal closure and certification premises. The differential contribution is this conditional interface among incompatibilist sourcehood, deployable prediction, unique projection, and standard search complexity.

Azadi more recently argues in the opposite explanatory direction, deriving conditions for autonomous behaviour from undecidability, computational irreducibility, and agent--environment coupling \cite{Azadi2025}. The present paper does not infer agency from computational hardness. It begins with an independently characterised substantive-agency model and asks what follows only when passive-extension and certification interfaces are added. This direction-of-explanation difference is essential: computational difficulty is neither A3-C nor A5 here.

\paragraph{Formal and process models of choice.}
Durham models a free choice as a singular process between macrostates and proposes a quantitative measure of adaptive freedom \cite{Durham2020}. That behavioural construction does not derive anticipatory non-pointing from singularising priority and does not connect a post-act certificate to standard search complexity. In philosophy, sourcehood, alternative possibilities, reasons-responsiveness, manipulation, and luck already form distinct axes of the debate \cite{Frankfurt1969,FischerRavizza1998,VanInwagen1983,Pereboom2001,Kane1996,Clarke2003,Mele2006}. The axioms below do not purport to replace that literature; they state exactly which commitments are required for the computational bridge.

List's recent account of free will in AI provides the closest contemporary contrast within the target philosophical domain \cite{List2025}. It treats intentional agency, alternative possibilities, and causal control as a pragmatic and explanatorily useful criterion, while explicitly declining to require indeterminism or unpredictability. The present paper does not dispute the usefulness of that functional notion. It asks what further structure is required by the stronger claim that the realised continuation is not fixed by causally sufficient pre-act conditions and that the act first supplies its exact point. The difference is therefore one of explanatory target: functional agency and control on List's account, substantive sourcehood and singularising priority here.

\paragraph{Sheaves, contextuality, and global sections.}
Sheaf-theoretic contextuality uses the obstruction to a compatible global section as a signature of contextuality \cite{AbramskyBrandenburger2011}. Our use is structurally different: a choice frame may possess several global sections. The issue is not non-existence of a global section but absence of an intrinsic or operationally available point selecting the section later actualised. Consequently, the sheaf language organises local compatibility and symmetry; it does not establish a complexity lower bound.

\paragraph{Total search and unique projection.}
The distinction between polynomial verification and polynomial search belongs to standard complexity theory, as do total search classes and search-to-decision constructions \cite{Papadimitriou1994,MegiddoPapadimitriou1991,AroraBarak2009}. The new claim is not that FNP or total search exists. It is that, under a separately justified agency model and certification protocol, unique projection makes any solver of the associated relation compute the very actualisation map excluded from FP by effective-pointing closure.

\begin{table}[t]
\centering
\small
\begin{tabular}{p{0.23\textwidth}p{0.29\textwidth}p{0.36\textwidth}}
\toprule
Literature & Established focus & Additional step supplied here \\
\midrule
Irreducibility and inference limits & No shortcut or universal predictor for specified process classes & Closed-frame pre-pointing criterion tied to substantive singularisation \\
Computational sourcehood & Structure-preserving simulation and functional sourcehood under determinism & Causal non-determination kept distinct from computational cost \\
Process models of choice & Singular transitions and behavioural measures of freedom & Axiomatic sourcehood, closure, and exact computational interfaces \\
Sheaf contextuality & Compatibility and obstruction to global sections & Several global sections with no pre-act point selecting the actual one \\
FNP and total search & Verification, totality, and algorithmic witness search & Historical certificates whose unique projection recovers actualisation \\
\bottomrule
\end{tabular}
\caption{Differential location of the present contribution. No row alone yields the main result; the contribution is the proved bridge across the rows.}
\label{tab:related}
\end{table}

The resulting novelty claim is therefore deliberately limited. We do not claim priority for computational unpredictability, sourcehood, sheaf models, or total search. We claim an explicit chain of independently auditable premises and theorems connecting substantive singularisation to a computability bifurcation and, under certified realisation, to a standard search-class consequence. The realisation-gap theorem is essential to that claim because it prevents the final certification premise from being smuggled into the earlier philosophical characterisation.

\section{Substantive agency: desiderata and formal core}\label{sec:axioms}

Let $Q$ be a set of occurrences. Each $q\in Q$ has a pre-act history $h(q)$ in a history space $H$. Let $G(q)$ be the set of globally coherent continuations compatible with that history. An actualisation map, defined on completed occurrences, assigns
\[
 \alpha(q)\in G(q).
\]
The notation records the realised continuation extensionally; it does not assume that $\alpha(q)$ was available or determined as a selected value before the act.

Before listing the constitutive clauses, we make explicit why anticipatory non-pointing is required by substantive singularisation rather than appended as a stronger computational condition.

\begin{definition}[Closed pre-act agency--information frame]\label{def:information-frame}
For an occurrence $q$, the declared agency--information system contains the deciding agent and every observer, device, or process informationally coupled to the decision situation before the act. Write $\mathcal P_{\rm act}(q)\subseteq G(q)$ for the continuations exactly pointed within that system at that occurrence. Its pre-act representation is \emph{actually access-complete} when every exact point operationally present anywhere within the declared system before the act belongs to $\mathcal P_{\rm act}(q)$. This closure concerns actual availability; it neither asserts that every true proposition is known nor turns the mere mathematical existence of an algorithm into an operationally present device.
\end{definition}

\begin{definition}[Substantive singularising priority]\label{def:singularising-priority}
An act has substantive singularising priority, relative to a closed agency--information frame, when it is the transition at which the realised continuation first becomes an exact point of that frame. Merely producing in history a continuation which was already exactly pointed within the frame is realisation, but not substantive singularisation by the act.
\end{definition}

The distinction rests on the act's \emph{differentiating contribution}. Realisation answers whether a continuation eventually occurs; substantive singularisation asks where the exact difference by which this continuation, rather than another admissible one, first enters the relevant frame. An act may remain causally necessary for the physical occurrence of a continuation already pointed in advance. It does not follow that the act also supplied the differentiating point. When that point was operationally available inside the closed pre-act system, the later occurrence implements or realises a prior distinction instead of originating that distinction at act level.

This claim does not identify knowledge with causation. A foreknowing process need not causally fix the act. It concerns the different question recorded by $\mathcal P_{\rm act}(q)$: whether the exact distinction among admissible continuations was already present and available before the occurrence. Nor does mere truth about the future suffice. A true but wholly inaccessible proposition supplies no operational point to the declared frame. The differentiating-contribution criterion therefore motivates anticipatory non-pointing without collapsing (A3-P) into causal non-determination (A3-C).

\begin{proposition}[Pre-pointing explication]\label{prop:prepointing-collapse}
In an actually access-complete pre-act agency--information frame, substantive singularising priority entails actual anticipatory non-pointing (A3-P).
\end{proposition}

\begin{proof}
Fix $q$ and suppose instead that $\alpha(q)\in\mathcal P_{\rm act}(q)$. By actual access-completeness, that continuation is already exactly pointed within the relevant system before the act. It is immaterial whether the point is held by the deciding organism or by a distinct but informationally coupled predictor: for the question of when the relevant system first acquires the distinction, they are components of the same closed frame. The subsequent act may realise the pointed continuation, but it cannot be the transition which first points it, contrary to substantive singularising priority. Conversely, foreknowledge wholly isolated from and operationally unavailable to the declared system is not an element of $\mathcal P_{\rm act}(q)$ and is not a counterexample.
\end{proof}

This proposition is an analytic explication of the temporal office assigned to the act, not the paper's technical lower-bound result. It concerns points actually present before an occurrence. The later passage from a merely existing polynomial algorithm to a deployable point requires a separate modal bridge.

\begin{definition}[Causal forcing semantics]\label{def:causal-forcing}
Fix a declared causal interpretation $\mathcal C$. For every history $h(q)$, let $\Omega_{\mathcal C}(q)$ be the nonempty set of nomologically admissible complete evolutions extending that history under the pre-act laws and mechanisms included in $\mathcal C$. Write
\[
 h(q)\Vdash_{\mathcal C}s
 \quad\Longleftrightarrow\quad
 \text{every }\omega\in\Omega_{\mathcal C}(q)\text{ continues with }s.
\]
Set $\mathcal D(q)=\{s\in G(q):h(q)\Vdash_{\mathcal C}s\}$. Since distinct continuations are mutually exclusive, $\mathcal D(q)$ is empty or a singleton. The causal interpretation, its system boundary, and its admissible evolutions must be fixed independently of the realised value $\alpha(q)$.
\end{definition}

This semantics does not solve the metaphysics of agent causation. It gives A3-C a precise model-relative meaning: the complete declared pre-act causal structure does not force the historical continuation across all its admissible evolutions. Sourcehood in A5 remains an additional attribution predicate and cannot be manufactured from non-determination alone.

An application must also justify the admissibility of its causal interpretation rather than select $\Omega_{\mathcal C}(q)$ opportunistically. At minimum, the declared interpretation should be complete relative to the physical or causal theory being used, stable under counterfactual variation of irrelevant details, and invariant under effective redescriptions which preserve the causal structure. Fixing $\mathcal C$ independently of $\alpha$ prevents direct ex post tailoring; these further adequacy conditions prevent causal openness from being manufactured merely by omitting a determining mechanism from the model.

\begin{definition}[Substantive agency frame]\label{def:saf}
A substantive agency frame is a tuple
\[
 \mathfrak A=(Q,H,h,G,\alpha,\mathcal D,\mathcal P_{\rm act})
\]
where $\mathcal D(q)$ is induced by a causal interpretation as in \cref{def:causal-forcing}, while $\mathcal P_{\rm act}(q)\subseteq G(q)$ contains the continuations exactly pointed and operationally present before that occurrence. These assignments are fixed independently of the completed graph of $\alpha$. Causal forcing is not defined by running time; actual pointing concerns informational presence rather than causal production. The frame satisfies the following requirements.
\begin{description}[style=nextline]
 \item[(A1) Local and global determination.] The history $h(q)$ imposes non-trivial constraints, and $G(q)$ contains precisely the continuations coherent with those constraints.
 \item[(A2) Admissible plurality.] $|G(q)|\ge 2$ for every occurrence in the relevant family.
 \item[(A3-C) Causal pre-act non-determination.] For every $q\in Q$, $\alpha(q)\notin\mathcal D(q)$. Thus no occurrence is counted as substantively open merely because a different member of the family lacks a uniform determining rule.
 \item[(A3-P) Actual anticipatory non-pointing.] For every $q\in Q$, $\alpha(q)\notin\mathcal P_{\rm act}(q)$. Thus the realised continuation is not already exactly pointed in the actual pre-act system. By \cref{prop:prepointing-collapse}, this is the necessary formal consequence of singularising priority in an actually access-complete frame.
 \item[(A4) Act-level actualisation.] The occurrence produces one definite continuation $\alpha(q)\in G(q)$.
 \item[(A5) Endogenous sourcehood.] The transition to $\alpha(q)$ is attributed to the agency represented by the frame, rather than to an exogenous intervention or to a merely independent randomiser.
\end{description}
\end{definition}

Conditions (A3-P) and (A4) jointly formalise substantive singularisation: the realised continuation is not an exact point actually present in the access-complete pre-act frame, and the act produces the definite post-act continuation. Thus (A3-P) is retained as an explicit constitutive clause for auditability, while (A4) is not burdened with a second copy of the same non-pointing requirement.

The assignments $\mathcal D(q)$ and $\mathcal P_{\rm act}(q)$ serve different purposes. An internal deterministic policy places its output in $\mathcal D(q)$ even when evaluation is super-polynomial. Computability of the completed graph does not by itself establish such causal provenance. An infallible observer may instead place the outcome in $\mathcal P_{\rm act}(q)$ without placing it in $\mathcal D(q)$. Pointwise formulation also blocks a loophole in which each occurrence is fixed or pointed by a different mechanism although no single global map covers the family.

Condition (A5) remains a philosophical interface axiom rather than a complete mathematical theory of agent causation: the choice-frame alone cannot manufacture endogenous authorship. Any application must supply an independently defensible agency locus and account of causal attribution. That locus need not be identified a priori with the ordinary biological individual; it may in principle be individual, distributed, collective, or transindividual. Changing it changes the frame and therefore requires the causal class $\mathcal D$ and the sourcehood attribution to be reassessed. Accordingly, the axioms provide a necessary-condition framework for conceptions committed to openness, sourcehood, and singularising priority; they do not prove that a candidate event satisfies A5.

\subsection{Five meanings of non-predetermination}

We distinguish five notions:
\begin{enumerate}[label=(\roman*)]
 \item \emph{causal non-determination}: no map grounded in the complete pre-act causal structure fixes the actual continuation;
 \item \emph{anticipatory non-pointing}: no exact point of the realised continuation is operationally available from the complete pre-act representation;
 \item \emph{structural non-pointing}: no continuation is distinguished by the symmetries and relations internal to the pre-act representation;
 \item \emph{computational non-anticipability}: $\alpha\notin\mathrm{FP}$ on the fixed uniform encoding;
 \item \emph{computability-level non-anticipability}: no Turing procedure computes $\alpha$ from the encoded histories.
\end{enumerate}

As bare predicates, no implication among these notions is automatic. A section can be structurally undistinguished yet externally selected by an arbitrary convention. A deterministic mechanism may be causally fixing but computationally expensive. An infallible predictor may point to an outcome without causing it. Conversely, a computable regularity in the completed occurrence map need not constitute either the causal mechanism or an actually available pre-act point. \Cref{prop:prepointing-collapse} derives actual non-pointing from singularising priority and actual access-completeness. Effective-pointing closure then adds the distinct modal bridge from exact efficient computation to admissible deployability and robust non-pointing.

\subsection{Openness and sourcehood are independent}

\begin{proposition}[Logical independence]\label{prop:independence}
Pre-act openness and endogenous sourcehood do not imply one another.
\end{proposition}

\begin{proof}
A physical randomiser can choose between two admissible continuations without the realised value being attributable to the agent; this supplies openness without sourcehood. Conversely, a deterministic reasons-responsive mechanism can count as the agent's own mechanism on a compatibilist account while its output is fixed by the complete prior state; this supplies sourcehood in a reduced sense without openness. Hence neither property entails the other.
\end{proof}

The conjunction of (A3-C) and (A5), not either property alone, excludes both causally deterministic unfolding and pure chance. Condition (A3-P) is not caused by (A3-C), but it is required by the different temporal claim that the act itself first points the realised continuation inside the closed frame \cite{Zagzebski1991}.

\begin{proposition}[Bare causality--pointing independence]\label{prop:causal-pointing-independence}
Absent substantive singularising priority and access-completeness, causal non-determination and anticipatory non-pointing do not imply one another.
\end{proposition}

\begin{proof}
An exact non-intervening observer can point to the realised continuation while no pre-act mechanism causes or fixes it; this satisfies the bare causal predicate (A3-C) while violating (A3-P). If that observer is informationally coupled to an access-complete decision frame, however, the violation also removes substantive singularising priority from the later act by \cref{prop:prepointing-collapse}. Conversely, an internal deterministic mechanism may fix the outcome but remain unavailable as an exact pre-act point in the declared operational regime, for example because its evaluation is not completed before the act; this can satisfy (A3-P) while violating (A3-C). The two predicates therefore remain causally distinct even though substantive singularisation requires both.
\end{proof}

\section{Countermodels and the role of each requirement}\label{sec:countermodels}

The following finite countermodels test the contribution of each requirement. They are diagnostic role tests, not a formal proof that every axiom is logically independent of the conjunction of all the others, and they do not exhaust every use of the expression ``free will''.

\begin{countermodel}[No plurality: deterministic automaton]
For each $q$, let $G(q)=\{s_q\}$. The act merely reveals or executes the unique coherent continuation. Conditions (A1), (A4), and a weak form of sourcehood may hold, but there is no alternative to singularise. This is deterministic agency, not substantive choice in the present sense.
\end{countermodel}

\begin{countermodel}[Plurality without coherence]
Let $G(q)$ be replaced by a list containing mutually incompatible or physically impossible descriptions. Numerical multiplicity alone then creates no genuine alternatives. Condition (A1) prevents this inflationary construction.
\end{countermodel}

\begin{countermodel}[Plurality already pointed]
Let $G(q)=\{s_0,s_1\}$ but include in $h(q)$ a bit $b$ and a rule selecting $s_b$. Both continuations remain members of the displayed set, yet one is already distinguished by the complete pre-act data. Formal plurality therefore does not entail openness.
\end{countermodel}

\begin{countermodel}[Open randomisation without sourcehood]
Let a causally independent fair coin select $s_0$ or $s_1$. The continuation is not fixed by the agent's prior reasons, but the event is attributable to the randomiser rather than to the agent. Randomness removes prediction without producing authorship.
\end{countermodel}

\begin{countermodel}[Sourcehood without openness]
Let a deterministic policy $f$ be constitutive of the agent and put $\alpha(q)=f(h(q))$. The action may express the agent's values and be free from external coercion. If $f$ is fixed by the complete prior state, however, the frame lacks (A3-C). This is the principal compatibilist countermodel.
\end{countermodel}

\begin{countermodel}[Computationally expensive determinism]
Let the same internal policy $f$ require super-polynomial time to evaluate. Practical or polynomial-time prediction may fail, but $f$ remains a causally licensed determination map in $\mathcal D$. The frame therefore still violates (A3-C). Computational expense alone does not convert deterministic unfolding into substantive openness.
\end{countermodel}

\begin{countermodel}[Non-causal infallible foreknowledge]
Suppose that the act is not causally fixed, but an observer exactly identifies $\alpha(q)$ before the act without intervening. If the observer is informationally coupled to the decision situation, actual access-completeness gives $\alpha(q)\in\mathcal P_{\rm act}(q)$: the frame is already pointed, so the later act lacks singularising priority. If the observer is wholly isolated, then $\alpha(q)\notin\mathcal P_{\rm act}(q)$ on that ground and the alleged foreknowledge does not violate (A3-P).
\end{countermodel}

\begin{countermodel}[No actualisation]
Let $G(q)$ remain non-singleton indefinitely and leave $\alpha$ undefined. The structure describes deliberative possibility but no completed act. Singularisation is required to pass from modal plurality to history.
\end{countermodel}

These examples reveal a dependency graph rather than a list of mutually exclusive properties. Plurality is a precondition for non-singularisation and singularisation. Singularisation must preserve coherence. Sourcehood determines how singularisation is interpreted. Causal openness and anticipatory non-pointing exclude different countermodels. None of these conditions by itself entails traceability or efficient verification.

\section{Compatibilism, individuation, and the sourcehood boundary}\label{sec:compatibilism}

The present section is not an autonomous philosophical excursus. Its function is to determine which prominent accounts retain the property required by the constitutive part of the argument: a source of the realised choice not wholly reducible to causally sufficient pre-act conditions. It does not yet establish polynomial non-anticipability or any complexity consequence.

\subsection{What compatibilism preserves}

Compatibilism is not refuted merely by repeating that a determined act was predictable. It is a family of views according to which freedom or moral responsibility can coexist with determinism, often by appealing to reasons-responsiveness, identification with motives, absence of certain coercive interventions, or ownership of the deliberative mechanism \cite{Frankfurt1969,FischerRavizza1998,SEPCompatibilism}. Frankfurt-style cases challenge the claim that alternative possibilities are necessary for moral responsibility \cite{Frankfurt1969}; Fischer and Ravizza's guidance control locates responsibility in the agent's appropriately reasons-responsive mechanism \cite{FischerRavizza1998}. List's criteria of intentional agency, alternative possibilities, and causal control extend a related functional approach to AI systems without making indeterminism or unpredictability constitutive \cite{List2025}. These accounts answer an important normative and practical question: when may conduct be attributed to a person or system for purposes of explanation and responsibility?

That achievement is substantial. A system which understands reasons, anticipates consequences, adjusts its conduct, and answers to norms can be an appropriate locus of praise, blame, prevention, and sanction. Nothing in the present paper shows such practices to be incoherent under determinism. The narrower question is what has thereby been preserved: functional control and attribution, or the stronger property of being the ultimate source of which admissible continuation becomes actual.

\begin{definition}[Functional attribution and ultimate sourcehood]\label{def:two-sourcehoods}
An act has \emph{functional attribution} when it issues through the reasons, dispositions, and control mechanisms assigned to a declared agent. It has \emph{ultimate sourcehood} in the present relative sense only when that attribution is conjoined with causal pre-act non-determination: the complete intrinsic pre-act conditions and governing mechanisms do not already suffice to select the realised continuation.
\end{definition}

The qualification ``in the present relative sense'' matters. The paper does not demand that an agent create itself or stand outside every causal condition. It asks whether, relative to the complete pre-act causal description declared by the model, the act contributes the selection rather than merely transmitting a selection already fixed there.

\subsection{The problem of ``the agent's own reasons''}

In a wholly deterministic universe, beliefs, preferences, character, and the mechanism which arbitrates among reasons may themselves result from prior states. An action may then issue authentically through the agent's psychology while remaining fixed by conditions of which that agent is not the ultimate source. The possessive in ``the agent's own reasons'' cannot settle this issue, because the relevant question is precisely what converts causal location or functional ownership into sourcehood.

Manipulation arguments make the gap vivid. An agent whose dispositions and deliberative architecture had been configured by an external manipulator might remain rational, reasons-responsive, and free from immediate coercion. The residual intuition that something is missing concerns provenance and sourcehood, not the mere interior location of the mechanism \cite{Pereboom2001}. Van Inwagen's consequence argument raises the related question whether, under determinism, consequences of the past and the laws can be within the agent's control \cite{VanInwagen1983}. Neither argument is treated here as a conclusive refutation of compatibilism. Together they show why an additional inference from internal control to ultimate sourcehood is required rather than automatic.

The distinction also blocks a tempting but invalid move. Computational expense cannot supply that missing inference. If an internal deterministic policy $f$ is causally sufficient for the act, then $f\in\mathcal D$ and (A3-C) fails whether $f$ runs in linear, polynomial, exponential, or unbounded practical time. An expensive determination is still a determination.

\subsection{The boundary of the agent}

Compatibilist attribution normally relies on a functional boundary: some mechanisms count as the agent's own, whereas coercive or manipulative causes are treated as external. That boundary is intelligible and often indispensable for explanation and law. It need not, however, be an ontologically primitive boundary.

Simondon's account reverses the order in which the individual and individuation are explained: the individual is a result of individuation rather than an unquestioned first term \cite{Simondon2005}. Foucault's archaeological analysis supplies a distinct epistemological caution: historically situated categories through which objects of knowledge are constituted should not automatically be treated as the ultimate articulations of reality \cite{Foucault1966}. Foucault does not provide a metaphysics of agency here, and Simondon does not decide the free-will dispute. Their joint methodological lesson is limited but relevant: the ordinary individual cannot simply be assumed to mark the point at which a determined causal chain changes its ontological nature.

The present definition is therefore \emph{boundary-neutral}, not invariant under arbitrary changes of boundary. It does not stipulate that the source must coincide with the ordinary human organism. It may in principle be located at an individual, distributed, collective, or transindividual scale. But every proposed locus defines a new substantive-agency frame, whose class $\mathcal D$ and attribution clause (A5) must be justified anew. Moving the boundary cannot by itself create causal non-determination.

This yields an asymmetry. Compatibilism needs a sufficiently stable agent boundary to distinguish ownership from external interference. Substantive agency first tests whether a declared source is reducible to sufficient prior determination; the social boundary used for imputation is then a further question. Criticising the ontological primacy of the individual therefore challenges a simple internal/external route to compatibilist ultimacy without, by itself, eliminating the substantive criterion.

\subsection{Practical responsibility and ultimate responsibility}

Practical responsibility identifies the appropriate locus for praise, blame, correction, deterrence, and answerability in a social order. Ultimate responsibility adds the stronger claim that the act expresses a source not wholly reducible to causally sufficient prior conditions. Compatibilism can provide a powerful account of the former while declining or failing to establish the latter. Strawson's emphasis on reactive attitudes illustrates how responsibility practices may be grounded without first resolving ultimate metaphysics \cite{Strawson1962}.

This distinction neither trivialises compatibilism nor turns existing legal practices into proof of substantive freedom. It identifies their evidential role correctly. The persistence and intelligibility of responsibility practices show the importance of agency and attribution; they do not deductively establish (A3-C). Conversely, denial of ultimate sourcehood does not logically force abandonment of every practical regime of responsibility.

\subsection{Progressive closure and the exact transmission boundary}

The principal alternatives can now be located by properties. Integral determinism supplies no ultimate sourcehood in the sense of \cref{def:two-sourcehoods}. Deterministic compatibilism preserves functional attribution and may preserve moral responsibility, but it satisfies the substantive target only by adding an account which secures causal pre-act non-determination. Pure indeterminacy avoids determination but, without (A5), supplies chance rather than agency. Libertarian theories differ over event-causal, agent-causal, and non-causal accounts of how openness and sourcehood combine \cite{Kane1996,OConnor2000,Clarke2003,Mele2006}. Hard incompatibilism rejects the required freedom without reducing that position to behavioural illusionism \cite{Pereboom2001}.

The resulting closure is therefore progressive rather than logically absolute. The paper does not prove that compatibilism and illusionism exhaust every conceivable philosophy. It establishes something narrower and more useful: among the positions considered, any view retaining real alternatives, ultimate sourcehood, and the act's priority in singularising one alternative must instantiate the substantive core; a refusal must abandon or reinterpret at least one of those properties. Available exact prescience supplies no intermediate substantive position: if it is inside the access-complete agency--information frame, it removes the act's singularising priority; if it is wholly outside that frame, it does not constitute a pre-act point for the model.

Most importantly, the philosophical result must not be converted into a computational conclusion by verbal substitution.

\begin{proposition}[Transmission boundary]\label{prop:transmission-boundary}
Non-reducibility of the realised choice to causally sufficient complete pre-act conditions, when represented in a substantive-agency frame, supplies condition (A3-C). Substantive singularising priority in an access-complete pre-act agency--information frame supplies condition (A3-P). These constitutive conclusions do not by themselves imply $\alpha\notin\mathrm{FP}$, Turing noncomputability, or certified actualisation.
\end{proposition}

\begin{proof}
The first statement is the pointwise causal content assigned to $\mathcal D(q)$. The second is \cref{prop:prepointing-collapse}: an actually available exact point would already point the closed pre-act object. This does not classify merely existing algorithms as actually available. Effective-pointing closure therefore introduces $\mathcal P_{\rm eff}(q)$ and the additional modal robustness needed to exclude FP. Finally, \cref{thm:realisation-gap} shows that computability outside FP does not guarantee a polynomially balanced, polynomial-time verifiable, uniquely projected trace relation.
\end{proof}

The terminological result can now be stated without hiding those bridges:
\begin{proposition}[Compatibilist boundary]\label{prop:compatibilist}
A determined but reasons-responsive act may satisfy a compatibilist conception of freedom and practical responsibility. It cannot satisfy substantive agency as defined in \cref{def:saf} if the complete intrinsic pre-act state and governing mechanisms already causally license a selector for the realised continuation.
\end{proposition}

This proposition does not settle which conception deserves the unqualified name ``free will''. It forces the disagreement to be stated at the level of properties rather than hidden in terminology. The later route to $\LAc$ remains exactly the formal route proved in \cref{sec:prediction,sec:objectification,sec:interfaces}: effective-pointing closure excludes FP, (O1)--(O6) provide the certified search interface, and (P1)--(P2) govern its historical interpretation.

\section{Choice-frame and sheaf-theoretic local-to-global representation}\label{sec:representation}

\subsection{Abstract representation}

For each occurrence $q$, define the unpointed pre-act choice object
\[
 C_q=(h(q),G(q)).
\]
After the act, the corresponding pointed object is
\[
 C_q^+=(h(q),G(q),\alpha(q)).
\]

\begin{proposition}[Choice-frame representation lemma]\label{thm:representation}
Every substantive agency frame determines a family of non-singleton causally undetermined and anticipatorily unpointed choice objects before the act and pointed choice objects after the act. Conversely, a family $(C_q,C_q^+)$ equipped with causal-determination and pre-act-pointing classes and satisfying coherence, non-singletonness, causal non-determination, anticipatory non-pointing, and endogenous pointing determines a substantive agency frame.
\end{proposition}

\begin{proof}
From a frame $\mathfrak A$, conditions (A1) and (A2) define $C_q$ as a coherent non-singleton possibility object. Condition (A3-C) gives $\alpha(q)\notin\mathcal D(q)$, while (A3-P) gives $\alpha(q)\notin\mathcal P_{\rm act}(q)$. Conditions (A4) and (A5) provide the endogenous post-act point and hence $C_q^+$. The converse simply reads these data back occurrence by occurrence.
\end{proof}

The lemma is intentionally elementary and is neither a classification theorem nor a complexity lower bound. The sheaf-theoretic refinement below formalises local compatibility and global continuation; no later lower bound depends on sheaf theory alone.

\subsection{Local-to-global realisation}

When a history is composed of overlapping local constraints, let $\mathcal U_q=\{U_i\}$ be a cover and $F_q$ a finite presheaf of locally admissible assignments. Compatible local sections glue to global continuations
\[
 \Gamma(F_q)=\{s:\text{$s$ satisfies every local and overlap constraint}\}.
\]
Set
\[
 \Gamma(F_q)=G(q).
\]
The actualisation operation is a distinct map
\[
 \alpha(q)\in G(q).
\]
Thus gluing answers ``which continuations are coherent?'', whereas singularisation answers ``which coherent continuation became actual?'' The companion-paper notation $\GLUE(x)$ and $\SELECT(q)$ is introduced at the computational interface, where an occurrence is written $q=(x,e)$. The distinction is structural and should not be mistaken for a lower bound.

This is a genuinely sheaf-theoretic use of the local-to-global method and belongs naturally to the topos-theoretic tradition in which compatible local data are organised by presheaves, sheaves, and their global sections \cite{MacLaneMoerdijk1992}. Its positive role is to separate $\GLUE$, which constructs the space of coherent global continuations, from $\SELECT$, which points to the continuation that becomes actual. The argument needs neither the full internal logic nor the general machinery of an ambient topos. In particular, the complexity conclusion does not follow from sheaf theory alone: it depends on the later modal and certified-actualisation bridges.

\subsection{Symmetry and intrinsic non-pointing}

Let
\[
 \mathcal G_q=\Aut(F_q,h(q))
\]
be the automorphisms preserving the intrinsic pre-act data. They act on $\Gamma(F_q)$. A strong form of structural non-singularisation is
\begin{equation}\label{eq:no-fixed-point}
 \Gamma(F_q)^{\mathcal G_q}=\varnothing.
\end{equation}

Under \eqref{eq:no-fixed-point}, no equivariant map from the one-point trivial $\mathcal G_q$-set to $\Gamma(F_q)$ exists. This excludes an intrinsic symmetry-preserving selector, but not every externally labelled or representation-dependent algorithm. That limitation becomes central in the next section.

\section{Prediction, selection, and the bridge principle}\label{sec:prediction}

\begin{definition}[Standard exact pre-act algorithm]\label{def:predictor}
Fix a uniform finite encoding of the complete declared pre-act input $q$. A standard exact pre-act algorithm is a uniform deterministic algorithm $M$ such that
\[
 M(q)=\alpha(q)
\]
on the entire occurrence family, using neither a post-act trace nor non-uniform advice encoding future outcomes. It is polynomial when its running time and output length are polynomial in $|q|$.
\end{definition}

This is the ordinary algorithmic notion needed for FP. It is deliberately not filtered by an additional philosophical extensionality test: a polynomial program on the fixed encoding belongs to the class even if it exploits canonical labels. Physical implementations may separately be required to be passive and non-intervening, but such requirements cannot remove a standard algorithm from FP.

At this stage, $\alpha\in\mathrm{FP}$ on the occurrence family means that one total uniform polynomial-time machine returns $\alpha(q)$ for every valid occurrence encoding $q$, with the polynomial bound measured in $|q|$. Once a polynomially decidable valid language $I$ is fixed in \cref{def:effective-branches}, the canonical value $\bot$ on invalid strings turns this family-relative computation into an ordinary total function on $\{0,1\}^*$.

\begin{lemma}[Algorithm-to-selector]\label{lem:predictor-selector}
Every standard exact pre-act algorithm induces a uniform selector of the realised continuation on the fixed encoding.
\end{lemma}

\begin{proof}
Define $\sigma_M(q)=M(q)$. Exactness gives $\sigma_M(q)=\alpha(q)\in G(q)$. Uniformity and the pre-act input restriction make $\sigma_M$ a pre-act rule rather than retrospective information. No causal conclusion follows from this lemma alone.
\end{proof}

The computational bridge is now factored into two premises whose content is independent of the conclusion $\alpha\notin\mathrm{FP}$.

\begin{definition}[Admissible passive extension]\label{def:passive-extension}
Given a frame $\mathfrak A$ and a uniform pre-act procedure $M$, an extension $\mathfrak A[M]$ is \emph{admissibly passive} when: (i) it preserves $q$, $h(q)$, $G(q)$, $\mathcal D(q)$, and the realised continuation $\alpha(q)$; (ii) $M$ receives only the complete pre-act encoding; (iii) its output is recorded before the act inside the enlarged agency--information frame but is not fed into the deciding mechanism; and (iv) it uses no post-act data, future-outcome table, oracle, or non-uniform advice. Thus the extension adds an epistemic point without adding a causal arrow into the act.
\end{definition}

\begin{definition}[Passive-extension invariance]\label{def:pei}
A substantive-agency family has \emph{passive-extension invariance} when every admissibly passive extension preserves the substantive status and singularising priority of each act. In other words, a causally inert read-only enlargement cannot by itself turn a substantively free act into a non-substantive one.
\end{definition}

\begin{definition}[Effective deployment principle]\label{def:edp}
Every standard polynomial-time procedure on the fixed complete encoding admits an implementation satisfying \cref{def:passive-extension}. This is an in-principle counterfactual implementation claim; it does not assert that the procedure is known or installed in the actual history.
\end{definition}

These premises expose the substantive dispute. Passive-extension invariance is motivated by causal irrelevance: if a shielded read-only computation neither changes the complete input nor feeds back into the agent, the freedom of the act should not depend on whether the computation is present. Yet singularising priority is informational as well as causal: an exact point inside the enlarged frame removes the act's status as the first pointing event. The incompatibility below makes that tension explicit rather than concealing it. The effective deployment principle then connects the standard extensional meaning of FP to a counterfactual physical implementation. It is stronger than polynomial computability alone: it abstracts from finite pre-act duration and other physical resource limits and requires preservation of the occurrence under installation. A reflexive theory may deny that an exact forecast can ever be added without changing the occurrence; a resource-sensitive theory may deny timely passive deployment; and a strongly epistemic theory may reject passive-extension invariance. None of these premises is equivalent by definition to $\alpha\notin\mathrm{FP}$.

\begin{theorem}[Modal incompatibility]\label{thm:modal-incompatibility}
The following three claims cannot hold jointly on the same occurrence family:
\begin{enumerate}[label=(\roman*)]
 \item the acts have substantive singularising priority;
 \item passive-extension invariance holds;
 \item an exact pre-act selector admits an admissibly passive extension.
\end{enumerate}
\end{theorem}

\begin{proof}
Suppose an exact procedure $M$ had such an extension. By exactness and clause (iii) of \cref{def:passive-extension}, the enlarged pre-act frame contains the point $M(q)=\alpha(q)$ before every act. The pre-pointing explication, \cref{prop:prepointing-collapse}, therefore says that the act in $\mathfrak A[M]$ cannot have singularising priority. Passive-extension invariance says that it retains precisely that priority because the extension is causally inert. Contradiction.
\end{proof}

\begin{corollary}[Passive-extension dilemma]\label{cor:passive-dilemma}
If an exact pre-act selector admits an admissibly passive extension, then an account of the occurrence must relinquish at least one of the following: substantive singularising priority or invariance of that priority under causally inert informational extension.
\end{corollary}

\begin{proof}
This is the contrapositive form of \cref{thm:modal-incompatibility} with the third claim fixed.
\end{proof}

The two branches have distinct content. Relinquishing singularising priority means that the exact point was present before the act, so the act realises rather than first differentiates the continuation. Relinquishing passive-extension invariance preserves first-pointing only by making the substantive status of the act sensitive to the presence of a read-only informational enlargement which, by stipulation, changes neither the causal history nor the realised continuation. The theorem does not decide that philosophical choice by definition; it establishes that exact passive anticipation carries one of these two costs. Effective-pointing closure selects the invariant branch and thereby converts this modal dilemma into an exclusion of efficient exact selectors.

\begin{definition}[Effective-pointing closure]\label{def:acp}
For brevity, \emph{effective-pointing closure} denotes the conjunction of passive-extension invariance and the effective deployment principle.
\end{definition}

\begin{corollary}[Efficient-selector exclusion]\label{thm:predictive-collapse}
If a substantive agency frame satisfies effective-pointing closure on a complete uniform encoding, then
\[
 \alpha\notin\mathrm{FP}.
\]
\end{corollary}

\begin{proof}
Suppose $\alpha\in\mathrm{FP}$. A standard exact polynomial-time algorithm $M$ then computes $\alpha(q)$ on the whole family. The effective deployment principle gives an admissibly passive extension for $M$, while passive-extension invariance preserves substantive singularising priority. This contradicts \cref{thm:modal-incompatibility}.
\end{proof}

The theorem and corollary separate two achievements. The theorem is an incompatibility result about singularising priority, causal invariance, and informational pre-pointing. The corollary excludes FP only after the additional deployment principle connects an abstract polynomial procedure to an admissibly passive pre-act implementation. The argument is conditional, but its conditions do not restate the computational conclusion. The polynomial threshold enters only through the deployment principle, not through the concept of causation or actual substantive agency.

\begin{proposition}[Robustness under effective representation]\label{prop:representation-robustness}
Let two encodings of the same occurrence family have polynomial-time intertranslations of valid inputs and continuations with polynomially related lengths, preserving the admissible sets and the uniquely projected certified actualisation as in the effective-equivalence definition of \cite{Clech2026a}. Then $\alpha\in\mathrm{FP}$ under one encoding if and only if $\alpha\in\mathrm{FP}$ under the other.
\end{proposition}

\begin{proof}
Compose a polynomial-time algorithm in either representation with the polynomial-time input translation and the inverse output translation. Polynomial length relations preserve a polynomial running-time bound.
\end{proof}

Representation robustness is therefore proved after fixing the standard algorithmic class; it is not used to discard representation-dependent polynomial algorithms from that class.

\section{Certified realisation as a standard search relation}\label{sec:objectification}

Two effective notions must be separated. A \emph{uniform finite presentation} encodes occurrences and continuations and thereby makes computability questions meaningful. A \emph{certified realisation} additionally gives short, uniformly verifiable post-act evidence. The first is enough for the exhaustive computable/noncomputable bifurcation. The second is needed for FNP.

\begin{definition}[Search classes]\label{def:search-classes}
A polynomially balanced relation $R(x,w)$ belongs to $\FNPS$ when membership is decidable in deterministic polynomial time; its search task is to output a witness on every input for which one exists. It belongs to $\TFNPS$ when it is additionally total, namely $\forall x\,\exists w\,R(x,w)$. A search relation belongs to $\FPS$ when one deterministic polynomial-time algorithm outputs an accepted witness on every input in its domain. Thus $\TFNPS\subseteq\FNPS$. Multiple witnesses may exist; unique projection below requires only that their continuation component agree.
\end{definition}

\begin{definition}[Certified-actualisation protocol]\label{def:objectification}
A substantive agency frame has a certified actualisation interface when the following conditions hold.
\begin{description}[style=nextline]
 \item[(O1) Finite uniform encoding.] Occurrences, histories, continuations, and traces have a fixed effective self-delimiting encoding.
 \item[(O2) Polynomially recognisable unbounded domain.] The valid occurrence encodings form an infinite language $I\subseteq\{0,1\}^*$ with unbounded input lengths and $I\in\mathrm P$.
 \item[(O3) Polynomial balance.] There is a fixed polynomial $p$ such that every accepted pair $w=(s,\tau)$ for every valid $q$ satisfies $|w|\le p(|q|)$.
 \item[(O4) Fixed polynomial verification and soundness.] A fixed deterministic polynomial-time predicate $V(q,s,\tau)$ verifies certificates, and acceptance for valid $q$ implies $s\in G(q)$.
 \item[(O5) Certified completeness and extensional unique projection.] Accepted certificates exist for every valid occurrence and every string accepted by the mathematical verifier projects to the same continuation:
 \[
   \exists\tau\,V(q,\alpha(q),\tau)=1,
   \qquad
   V(q,s,\tau)=1\Rightarrow s=\alpha(q).
 \]
 \item[(O6) Pre-act completeness and occurrence discipline.] Each occurrence input is $q=(x,e)$, where $x$ contains the intrinsic information declared relevant before the act and $e$ is a polynomially bounded occurrence identifier fixed independently of the realised continuation. Prior admissibility depends only on $x$, so $G(x,e)=G(x,e')$ whenever both occurrences share the same prior state. Neither a post-act trace nor outcome-encoding advice occurs in $q$.
\end{description}
\end{definition}

Condition (O5) is an independent certified-actualisation axiom. It is not entailed by substantive agency, computability, the physical existence of a post-act trace, ordinary authentication, or polynomial verification. In particular, verification of a digital signature establishes validity relative to a key but does not by itself ensure that no mathematically valid signature string exists for an alternative continuation. A concrete application must therefore specify how the verifier's extensional accepted language is bound to one committed post-act state. The results below assume this property; they do not derive it from generic recording technology.

\begin{definition}[Non-oracular provenance condition]\label{def:provenance}
For the extensional relation to represent historical recording rather than a stipulated answer graph, require:
\begin{description}[style=nextline]
 \item[(P1) Post-act generation.] A fixed outcome-independent recording channel $\rho$, committed before the act, produces its raw record only after actualisation. A fixed polynomial-time extractor converts that record into $\tau$. Neither $\rho$, the extractor, nor $V$ contains occurrence-specific advice, a future-outcome table, or an oracle for $\alpha$.
 \item[(P2) Verifiable linkage.] The verifier checks that the trace is linked to the precommitted occurrence identifier and to the record emitted by $\rho$. Any cryptographic or physical trust assumption used for that linkage must be stated as part of the concrete application.
\end{description}
\end{definition}

Define
\begin{equation}\label{eq:rsel}
 \RSEL(q,s,\tau)=1
 \quad\Longleftrightarrow\quad V(q,s,\tau)=1.
\end{equation}
The associated search witness is the pair $w=(s,\tau)$. Conditions (O1)--(O6) define the extensional search problem. Conditions (P1)--(P2) are deliberately not used to prove its complexity classification: they state the additional scientific obligation for interpreting accepted strings as records caused by completed acts rather than as a disguised outcome oracle. A precommitted authenticated append-only recorder may support provenance by binding $e$ before the act and appending $(e,s,r)$ afterwards, but authentication alone does not establish (O5). The extensional construction is totalised explicitly by
\[
 \widehat V(q,w)=1
 \quad\Longleftrightarrow\quad
 \bigl(q\in I\land w=(s,\tau)\land V(q,s,\tau)=1\bigr)
 \ \lor\ 
 \bigl(q\notin I\land w=\bot\bigr).
\]
Thus $\bot$ is accepted exactly on invalid inputs and rejected on valid ones; membership remains polynomial-time decidable by (O2).

\begin{proposition}[Certified realisation entails computability]\label{prop:objectification-computable}
Every family satisfying (O1)--(O5) has a Turing-computable singularisation map $q\mapsto\alpha(q)$ on its valid domain.
\end{proposition}

\begin{proof}
On input $q$, enumerate all strings $w=(s,\tau)$ of length at most the polynomial bound in (O3) and evaluate the fixed verifier from (O4). Historical adequacy guarantees that at least one string is accepted. Unique projection guarantees that every accepted string has the same first component, namely $\alpha(q)$. Returning that component computes $\alpha(q)$. The procedure may take exponential time; the proposition asserts computability, not efficient computability.
\end{proof}

The converse fails. This fact is essential: it prevents the exhaustive philosophical bifurcation from silently containing the FNP hypothesis needed for the complexity result.

\begin{theorem}[Realisation gap]\label{thm:realisation-gap}
There exists a total Boolean function $f:\{0,1\}^*\to\{0,1\}$ that is Turing-computable and not polynomial-time computable, but for which no polynomially balanced polynomial-time decidable total relation has unique projection $f$.
\end{theorem}

\begin{proof}
By the deterministic time-hierarchy theorem, choose a decidable language $L\notin\mathrm{EXP}$ and let $f(x)=1$ exactly when $x\in L$ \cite{AroraBarak2009}. Suppose that a relation $R(x,b,\tau)$ were decidable in polynomial time, total, polynomially balanced, and satisfied
\[
 R(x,b,\tau)=1\quad\Longrightarrow\quad b=f(x).
\]
On input $x$, enumerate every pair $(b,\tau)$ within the polynomial balance and run the verifier. Totality finds an accepted pair, and unique projection returns $f(x)$. There are at most $2^{\poly(|x|)}$ candidates, so this decides $L$ in exponential time, contradicting $L\notin\mathrm{EXP}$.
\end{proof}

Thus a computable but non-polynomial singularisation need not possess short retrospectively checkable certificates. The computational companion defines and analyses the corresponding certified families, but does not prove that a natural human-decision family satisfies them. The present theorem shows why such satisfaction is substantive interface work, not a consequence of computability alone.

\begin{remark}[Temporal and extensional viewpoints]\label{rem:temporal}
Complexity theory treats $\RSEL$ extensionally as a fixed relation. Conditions (P1)--(P2) do not strengthen the class-theoretic proof; they govern whether that relation is a faithful historical model. A concrete application must instantiate the channel, linkage, and trust assumptions rather than merely postulate accepted strings.
\end{remark}

\section{Exhaustive bifurcation and companion interfaces}\label{sec:interfaces}

\subsection{Notation concordance}

The three papers use the following common interface without identifying distinct levels:
\begin{equation}\label{eq:notation-concordance}
 \begin{gathered}
 q=(x,e),\qquad G(q)=\Gamma(F_q)=\GLUE(x),\qquad
 \alpha(q)=\SELECT(q),\\
 \RSEL(q,s,\tau)=1\ \Longleftrightarrow\ V(q,s,\tau)=1,
 \qquad w=(s,\tau).
 \end{gathered}
\end{equation}
For the sequential ontological presentation, the complete occurrence input $q_i$ is effectively identified with the causal frame $\langle i,h_i\rangle$, and $s_i^*=\alpha(q_i)$. The symbols $\LAc$, $\LAo$, and $\LAo^K$ retain exactly the meanings fixed in their respective companion papers. Only $\LAeff$ and $\LAnc$ are new intermediate branch names introduced here.

\subsection{The level at which exhaustiveness holds}

\begin{definition}[Effective presentation and the two branches]\label{def:effective-branches}
An agency family is \emph{effectively presented} when its valid occurrences form an infinite unbounded language $I\subseteq\{0,1\}^*$ with $I\in\mathrm P$, its complete declared pre-act inputs and continuations have uniform finite encodings, and $\alpha$ is total on $I$. Extend it canonically by $\bar\alpha(q)=\alpha(q)$ for $q\in I$ and $\bar\alpha(q)=\bot$ otherwise. A substantive agency frame realises \emph{effective substantive agency}, denoted $\LAeff$, when $\bar\alpha$ is Turing-computable and $\bar\alpha\notin\mathrm{FP}$ on the fixed encoding. It realises the \emph{noncomputable branch}, denoted $\LAnc$, when $\bar\alpha$ is not Turing-computable from the complete encoded pre-act inputs.
\end{definition}

The adjective ``effective'' in $\LAeff$ describes the computability of the completed occurrence map, not causal determinism, efficient computation, or FNP certification. The notation $\LAnc$ is deliberately distinct from $\LAo$: the ontological companion uses $\LAo$ for global plurality, structural non-singularisation, and pointing by the act, and $\LAo^K=\LAo+\AIK$ for its strong incompressible model. The two papers therefore do not assign different meanings to the same symbol.

\begin{theorem}[Effective--noncomputable bifurcation]\label{thm:bifurcation}
Let $\mathfrak A$ be an effectively presented substantive agency frame satisfying effective-pointing closure. Exactly one of the following holds:
\begin{enumerate}[label=(\alph*)]
 \item $\mathfrak A$ realises $\LAeff$;
 \item $\mathfrak A$ realises $\LAnc$.
\end{enumerate}
\end{theorem}

\begin{proof}
The totalised encoded map $\bar\alpha$ is either Turing-computable or not. These cases are mutually exclusive and exhaustive. In the computable case, \cref{thm:predictive-collapse} excludes every uniform polynomial-time algorithm agreeing with $\alpha$ on all valid occurrences. If $\bar\alpha$ belonged to $\mathrm{FP}$, its restriction to $I$ would be such an algorithm, a contradiction; hence $\bar\alpha\notin\mathrm{FP}$ and the frame realises $\LAeff$. In the noncomputable case, the definition gives $\LAnc$.
\end{proof}

In compressed form, the theorem establishes
\begin{equation}\label{eq:central-bifurcation}
 \begin{gathered}
 \text{substantive agency}+\text{effective presentation}\\
 {}+\text{effective-pointing closure}
 \quad\Longrightarrow\quad
 \LAeff\ \mathbin{\dot\vee}\ \LAnc
 \end{gathered}
\end{equation}
where $\dot\vee$ denotes mutually exclusive alternatives. The display is mnemonic rather than a replacement for the theorem's hypotheses. The substantive premise includes (A3-C) independently of running time; this prevents deterministic super-polynomial mechanisms from entering $\LAeff$. The FP exclusion comes from the separately stated modal effective-pointing closure, not from actual non-pointing alone.

This is a domain-relative exhaustive classification, not a closure of every philosophical or physical alternative. Inside the declared domain there is no third computability status between computable and noncomputable. Outside it, one may reject effective presentation, passive-extension invariance, or effective deployment without thereby adopting compatibilism or illusionism. There is also a distinct \emph{realisation question}: whether an $\LAeff$ family has the certified-actualisation interface required by the computational companion. \Cref{thm:realisation-gap} proves that this further property is not automatic.

\subsection{Computational interface}

\begin{definition}[Computational free will]
A family realises computational free will, denoted $\LAc$, exactly in the sense of \cite{Clech2026a}: for occurrence inputs $q=(x,e)$ it satisfies the six clauses $(D,P,G,S,T,A)$ stated there. The notation is not strengthened here. Conditions (O1)--(O6) are the formal entry hypotheses; (P1)--(P2) are the additional provenance conditions for a historical interpretation.
\end{definition}

\begin{theorem}[Interface to computational global selection]\label{thm:computational-interface}
Let $\mathfrak A$ realise $\LAeff$ and satisfy (O1)--(O6). Write each occurrence as $q=(x,e)$ as in (O6), and identify
\[
 \GLUE(x)=G(q),\qquad \SELECT(q)=\alpha(q).
\]
Then $\mathfrak A$ realises $\LAc$ in the exact sense of \cite{Clech2026a} and satisfies the additional entry hypotheses of its conditional separation theorem.
\end{theorem}

\begin{proof}
The representation lemma and (O6) provide clause (D). Conditions (A1) and (A2) give the coherent non-singleton set $\GLUE(x)=G(q)$ required by (P) and (G), while (A4) gives (S) with $\SELECT(q)=\alpha(q)$. Conditions (O3)--(O5) give the uniform retrospective trace required by (T). The definition of $\LAeff$ supplies $\alpha\notin\mathrm{FP}$, and \cref{prop:representation-robustness} gives the representation-robust form of (A). Thus $(D,P,G,S,T,A)$ hold. Conditions (P1)--(P2), when also instantiated, justify interpreting the accepted traces as post-act historical records.
\end{proof}

\begin{proposition}[Certified graph interface]\label{prop:certified-graph}
Under (O1)--(O5), the accepted witnesses form a polynomially balanced, polynomial-time decidable relation which is total on valid occurrences and whose continuation projection is the graph of $\alpha$: for every valid $q$, accepted witnesses exist, and every accepted $(s,\tau)$ satisfies $s=\alpha(q)$.
\end{proposition}

\begin{proof}
Polynomial balance and verification are (O3) and (O4). Existence and unique continuation projection are the two clauses of (O5). Hence the set of projected accepted pairs is exactly $\{(q,\alpha(q)):q\in I\}$, although several traces $\tau$ may certify the same projected value.
\end{proof}

This proposition identifies the constructive computational interface without attributing it to substantive agency alone. In complexity-theoretic terms, (O3)--(O5) supply a total polynomial-witness presentation of the graph of $\alpha$ with unique output projection. The later search-to-decision implication is standard. The distinctive result of the paper lies in composing that independently stated interface with the modal route from substantive singularisation to exclusion of an FP selector, while \cref{thm:realisation-gap} proves that neither side silently contains the other.

\begin{lemma}[Certified search relation]\label{lem:certified-search}
If (O1)--(O5) hold, the totalised search problem associated with $\RSEL$ belongs to $\TFNPS$. If, in addition, $\alpha\notin\mathrm{FP}$ on valid occurrences, it does not belong to $\FPS$.
\end{lemma}

\begin{proof}
Polynomial balance, polynomial verification, and historical completeness give an FNP search relation on the valid language. Since $I\in\mathrm P$, invalid inputs can be totalised by the fixed witness $\bot$ without changing polynomial verification or balance; the resulting relation lies in $\TFNPS$. Suppose a deterministic polynomial-time search algorithm returned some accepted witness $(s,\tau)$ on every valid $q$. Unique projection in (O5) forces $s=\alpha(q)$. Projecting the first component would therefore compute $\alpha$ in polynomial time on the entire valid family, contrary to $\alpha\notin\mathrm{FP}$. Hence the relation is not in $\FPS$.
\end{proof}

\begin{theorem}[Substantive agency to conditional complexity separation]\label{thm:conditional-separation}
Let $\mathfrak A$ be a substantive agency frame satisfying effective-pointing closure and the certified-actualisation protocol (O1)--(O6). Then $\mathfrak A$ realises $\LAc$, its associated total search relation lies in $\TFNPS\setminus\FPS\subseteq\FNPS\setminus\FPS$, and therefore
\[
 P\ne NP.
\]
This is conditional on the existence of a family satisfying the stated agency, closure, effective-presentation, and certification premises.
\end{theorem}

\begin{proof}
Conditions (O1), (O2), and (O6) supply an effective presentation. Conditions (O1)--(O5) make $\alpha$ Turing-computable by \cref{prop:objectification-computable}, so the noncomputable branch is excluded. Effective-pointing closure and \cref{thm:predictive-collapse} give $\alpha\notin\mathrm{FP}$. Hence the frame realises $\LAeff$, and \cref{thm:computational-interface} yields $\LAc$. By \cref{lem:certified-search}, the associated total relation belongs to $\TFNPS\setminus\FPS$.

It remains to record the standard decision consequence. Assume for contradiction that $P=NP$. Replace each witness by a self-delimiting encoding padded to exactly $p(|x|)$ bits; the verifier checks the delimiter, ignores valid padding, and retains polynomial time. For any such polynomially balanced polynomial-time relation $R(x,w)$, the language
\[
 L_R=\{(x,u):\text{some accepted witness }w\text{ extends the prefix }u\}
\]
is in NP and hence, under the assumption, in P. Starting with the empty prefix, query whether an accepted length-$p(|x|)$ completion exists after appending $0$; if so retain $0$, otherwise retain $1$. After exactly $p(|x|)$ queries the decoded string is an accepted witness. Thus, under $P=NP$, every FNP search relation, and in particular every TFNP relation, has a deterministic polynomial-time selector. This contradicts the relation just constructed. Therefore $P\ne NP$.
\end{proof}

\begin{corollary}[Historically grounded realisation]\label{cor:historical-grounding}
Under the hypotheses of \cref{thm:conditional-separation}, if (P1)--(P2) also hold, the certificates used by $\RSEL$ are generated and linked post-act without a future-outcome oracle. The complexity conclusion is unchanged; the additional conditions justify the intended historical interpretation.
\end{corollary}

The proof is included to make the paper logically autonomous. The computational preprint \cite{Clech2026a} supplies the fuller sheaf-to-search development and encoding conventions, but its acceptance or publication is not a premise of \cref{thm:conditional-separation}.

\subsection{Ontological interface}

\begin{theorem}[Interface to the ontological companion]\label{thm:ontological-interface}
Let $\mathfrak A$ realise $\LAnc$, and consider a sequential presentation in which each complete occurrence input $q_i$ is effectively equivalent to the causal frame $\langle i,h_i\rangle$ of \cite{Clech2026b}. Put $s_i^*=\alpha(q_i)$. Then no uniform Turing program computes $s_i^*$ from the causal frames throughout the family. If, in addition, the sheaf representation satisfies structural non-singularisation, then its global plurality, unpointed pre-act object, and act-level pointing instantiate the structural core denoted $\LAo$ in the ontological companion.
\end{theorem}

\begin{proof}
By definition of $\LAnc$, no Turing-computable function selects $\alpha$ uniformly from the complete encoded pre-act inputs. Effective equivalence of $q_i$ with $\langle i,h_i\rangle$ transfers that conclusion to the sequential presentation. Under structural non-singularisation, \cref{thm:representation} and the sheaf realisation provide exactly the plurality, absence of an intrinsic symmetry-preserving selector, and pointing by the act used to define $\LAo$ in the ontological companion. Neither statement implies its stronger incompressibility axiom.
\end{proof}

\begin{corollary}[Imported ontological consequence]\label{cor:ontological-companion}
If the additional structural hypothesis of \cref{thm:ontological-interface}, the causal algorithmic incompressibility axiom $\AIK$, and the admissible reconstructive trace protocol of \cite{Clech2026b} hold, then the family realises that paper's strong model $\LAo^K$ and its information-adjunction and incompressibility results apply. The $\LAnc$ condition alone does not establish near-maximal Kolmogorov incompressibility \cite{Kolmogorov1965,LiVitanyi2019}.
\end{corollary}

\subsection{Domain-relative classification}

\begin{proposition}[Effective-domain classification]\label{thm:trilemma}
Let a candidate family have an effective presentation. If it fails either substantive agency or effective-pointing closure, case (1) records that failure relative to the theorem's declared domain. Otherwise exactly one of cases (2) and (3) holds:
\begin{enumerate}[label=(\arabic*)]
 \item the family does not jointly instantiate substantive agency, passive-extension invariance, and effective deployment;
 \item it instantiates substantive agency and realises $\LAeff$;
 \item it instantiates substantive agency and realises $\LAnc$.
\end{enumerate}
\end{proposition}

The classification is exhaustive only inside its explicitly declared domain and is a logical partition rather than an independent lower-bound theorem. Failure of effective presentation lies outside it; failure of substantive agency, passive-extension invariance, or effective deployment lies in case (1), not in a hidden fourth computability status. Actual non-pointing follows from singularising priority and actual access-completeness; the counterfactual obstruction follows separately from \cref{thm:modal-incompatibility}. Branch (2) together with (O1)--(O6) yields $\LAc$ and activates the $P\ne NP$ consequence. Conditions O1--O5 exclude branch (3). Failure of certified actualisation is failure of an additional scientific interface.

\section{Scope, correspondence, and points of refusal}\label{sec:scope}

\subsection{Exact versus statistical anticipation}

The argument concerns exact uniform selection over an unbounded family. It is compatible with:
\begin{itemize}
 \item high-probability prediction;
 \item prediction on finitely many or typical occurrences;
 \item heuristics effective on practically relevant instance sizes;
 \item bounded rationality and incomplete pre-act models;
 \item deterministic chaos that is difficult to calculate in practice;
 \item stochastic models that predict distributions rather than realised choices.
\end{itemize}

Chaos and computational expense can produce practical unpredictability without substantive openness. Conversely, physical randomness can defeat exact prediction without sourcehood. The target is narrower and stronger: an exact, passive, uniform procedure which, from the declared complete pre-act input, selects the continuation that the agent will actualise throughout an unbounded family.

\subsection{Human correspondence and empirical limits}

The formal results apply to encoded families, not directly to human beings. To connect them to human agency one needs a correspondence principle of the following form.

\begin{definition}[Human correspondence principle]\label{def:hcp}
A class of human decision processes is adequately modelled by a substantive agency frame when the specified histories capture the information relevant to the decision problem, the admissible continuations correspond to genuine action alternatives, the actualisation map corresponds to completed acts, anticipatory availability is assessed in a declared operational regime, and sourcehood is supported by an independently defensible theory of agency.
\end{definition}

No finite behavioural dataset can establish the absence of a polynomial algorithm over an infinite family. Empirical evidence can support the adequacy of the representation, defeat particular predictors, or constrain proposed mechanisms. It cannot by itself prove an asymptotic lower bound. Likewise, a philosophical argument can justify why exact efficient anticipation would undermine substantive openness without proving that human neural or physical processes realise a particular complexity class.

The interdisciplinary contribution is therefore a conditional transfer theorem, not a new generic lower-bound technique or a reduction of the human to a machine model. If a human or non-human family instantiates effectively presented substantive agency, passive-extension invariance, and effective deployment, \cref{thm:bifurcation} forces it into $\LAeff$ or $\LAnc$. Rejection must be located precisely: causal forcing, sourcehood, singularising priority, actual access-completeness, passive-extension invariance, effective deployment, or effective presentation. Failure of certification blocks the route from $\LAeff$ to $\LAc$; failure of (P1)--(P2) blocks its historical interpretation. Conversely, once (O1)--(O6) hold, \cref{thm:conditional-separation} excludes $\LAnc$ and establishes the conditional computational consequence.

The theorem does not establish that compatibilism and illusionism are the only philosophical positions outside the model. Source-incompatibilist views rejecting alternative possibilities or theories denying finite effective presentation may reject a declared premise. An account claiming operationally available exact foreknowledge while retaining substantive singularisation must instead deny either informational closure or the priority claim that the act first points the continuation; it cannot retain all three. The paper's claim is conditional but constraining: the location and computational consequence of each refusal are made explicit.

\subsection{Limits and exact points of refusal}\label{sec:objections}

The architecture makes each principled point of refusal explicit and assigns it a precise consequence.

\paragraph{Causality, prescience, and pointing.} Condition (A3-C) now has the forcing semantics of \cref{def:causal-forcing}. Condition (A3-P) concerns points actually present. The modal result is explicitly an incompatibility theorem: singularising priority, passive-extension invariance, and a passively deployable exact selector cannot coexist. The FP corollary additionally requires effective deployment. Either bridge premise remains contestable, especially for reflexive or physically resource-bounded agents.

\paragraph{Standard algorithms and representation.} Equation \eqref{eq:no-fixed-point} excludes only equivariant intrinsic selectors. The FP conclusion instead follows from effective-pointing closure applied to every standard polynomial algorithm on the fixed encoding. \Cref{prop:representation-robustness} separately prevents the result from depending on polynomially equivalent presentations.

\paragraph{Certified actualisation.} The formal search consequence uses O1--O6. Condition O5 is an independent extensional unique-projection axiom; it is not supplied by ordinary recording, signatures, or authentication. Conditions (P1)--(P2) separately require a precommitted post-act channel and verifiable linkage, excluding future-answer tables, advice, and oracles. They do not strengthen the extensional proof; they determine whether it represents historical actualisation. \Cref{thm:realisation-gap} shows that certification cannot be inferred from computability.

\paragraph{Sourcehood.} Excluding an independent randomiser does not provide a complete mathematical theory of endogenous agency. A5 is explicitly an interpretive interface and any human application requires a separate theory of causal attribution.

\paragraph{Existence and exhaustiveness.} The $\LAeff/\LAnc$ dichotomy is exhaustive only for effectively presented substantive agency under effective-pointing closure. Neither this paper nor the companions prove that a natural human family satisfies all clauses of $\LAc$. The new corollary proves the universal raccord for every family that does satisfy them; it does not supply the existential or human-correspondence premise.

\paragraph{Search conventions.} \Cref{def:search-classes} distinguishes $\FPS$, $\FNPS$, and the standard total subclass $\TFNPS$. O1--O6 include the decidable domain, balance, and input discipline; the proof of \cref{thm:conditional-separation} gives the padded prefix-search reduction explicitly.

These limits remain explicit because the programme's force depends on keeping causal characterisation, computational status, certification, and human correspondence logically separate.

\section{Conclusion}

The paper establishes two principal results. First, substantive singularising priority, passive-extension invariance, and passive deployment of an exact pre-act selector are jointly incompatible. This yields a substantive dilemma: exact passive anticipation requires either loss of the act's first-pointing office or acceptance that its substantive status changes under a causally inert informational enlargement. Under the independently stated effective-deployment principle, choosing invariance excludes every FP selector of the actualisation map. Second, when completed acts additionally admit polynomially bounded, polynomially verifiable certificates with extensional unique projection, the certified graph interface turns historical actualisation into a total standard search relation. The associated relation lies in $\TFNPS\setminus\FPS$, and a self-contained search argument conditionally yields $P\ne NP$.

The route to those results is deliberately stratified. Coherent plurality, causal non-determination, anticipatory non-pointing, act-level singularisation, and endogenous sourcehood constitute the substantive core; no complexity assumption enters that core. Choice frames and sheaves represent the transition from compatible local data to several global continuations and then to a post-act point. They make the distinction between $\GLUE$ and $\SELECT$ precise, while the lower bound itself comes only from the modal and certification bridges. On an effective presentation, the actualisation map is either computable outside FP, yielding $\LAeff$, or noncomputable, yielding $\LAnc$.

The remaining questions are realisation questions rather than gaps hidden inside the proofs. The paper does not establish that human beings instantiate effective-pointing closure or certified actualisation. The realisation-gap theorem shows, moreover, that $\LAeff$ alone does not entail polynomially verifiable traces. If (O1)--(O6) hold, unique projection excludes the noncomputable branch and activates the conditional separation theorem; (P1)--(P2) justify the intended non-oracular historical interpretation. If the noncomputable branch holds instead, the ontological interface supplies uniform noncomputability, while near-maximal incompressibility still requires the additional axiom $\AIK$. The contribution is therefore a precise conditional architecture: each route to the complexity conclusion, and each principled point at which it may be refused, is explicitly identified.

\paragraph*{Acknowledgements.}
Generative AI tools were used solely for language editing, \LaTeX{} preparation, formatting, and mechanical consistency checks. The author independently developed and verified all conceptual and mathematical content and assumes full responsibility for the manuscript.

\bibliographystyle{plain}
\bibliography{references}

\end{document}